\documentclass{l4dc2025}

\newtheorem{assumption}{Assumption}

\newcommand{\OD}{\mathsf{D}}
\newcommand{\OF}{\mathsf{F}}
\newcommand{\OG}{\mathsf{G}}

\newcommand{\real}{\mathbb{R}}

\title[Neural Network-assisted Interval Reachability]{Neural Network-assisted Interval Reachability  for Systems with  Control Barrier Function-Based Safe Controllers}
\usepackage{times}

\author{%
 \Name{Damola Ajeyemi}$^*$ \Email{dajeyemi@bu.edu}\\
 \addr Division of Systems Engineering, Boston University
\AND
 \Name{Saber Jafarpour}\thanks{These authors contributed equally.} \Email{saber.jafarpour@colorado.edu}\\
 \addr Department of Electrical, Computer, and Energy Engineering, University of Colorado Boulder %
\AND
 \Name{Emiliano Dall'Anese} \Email{edallane@bu.edu}\\
 \addr Department of Electrical and Computer Engineering and Division of Systems Engineering, Boston University%
 }

\begin{document}

\maketitle

\begin{abstract}
Control Barrier Functions (CBFs) have been widely utilized in the design of optimization-based controllers and filters for dynamical systems to ensure forward invariance of a given set of safe states. 
While CBF-based controllers offer safety guarantees, they can compromise the performance of the system, leading to undesirable behaviors such as unbounded trajectories and emergence of locally stable spurious equilibria. 
Computing reachable sets for systems with CBF-based controllers is an effective approach for runtime performance and stability verification, and can potentially serve as a tool for trajectory re-planning. 
In this paper, we propose a computationally efficient interval reachability method for performance verification of systems with optimization-based controllers by: (i) approximating the optimization-based controller by a pre-trained neural network to avoid solving optimization problems repeatedly, and (ii) using mixed monotone theory to construct an embedding system that leverages state-of-the-art neural network verification algorithms for bounding the output of the neural network. Results in terms of closeness of solutions of trajectories of the system with the optimization-based controller and the neural network are derived. Using a single trajectory of the embedding system along with our closeness of solutions result, we obtain an over-approximation of the reachable set of the system with optimization-based controllers. Numerical results are presented to corroborate the technical findings. 
\end{abstract}


\section{Introduction}
\label{sec:introduction}

Control Barrier Functions (CBFs) have emerged as a powerful tool for the design of optimization-based control methods for safety-critical systems. Given a desired set of \emph{safe} states, CBF-based constraints
can be embedded into the optimization problem defining the control law to enforce forward invariance of the safe set (see, e.g., \cite{PW-FA:07,ADA-SC-ME-GN-KS-PT:19,xiao2023safe,garg2024advances}). CBFs have been widely applied in the context of safety filters, where a nominal controller is minimally modified to satisfy CBF constraints. They have also been used alongside control Lyapunov functions to achieve both safety and stability (\cite{ames2014control,ong2019universal}), and to address scenarios involving high-relative-degree constraints (\cite{xiao2019control}). The optimization problem defining the controller can also include input constraints   (\cite{agrawal2021safe,chen2024equilibria,cortez2021robust}).

While CBF-based controllers and safety filters ensure forward invariance of the safe set, they can impact the performance of the system by introducing undesirable behaviors. For instance, it is well-known that designing CBF filters for systems with stabilizing nominal controllers can lead to emergence of \emph{undesirable} equilibrium points or unbounded trajectories. Moreover, some of these undesirable equilibria may even be locally stable (see, e.g.,~\cite{MFR-APA-PT:21,WSC-DVD:22-tac,tan2024undesired,chen2024characterization}), and their stability properties cannot be changed by simply changing the CBF (\cite{chen2024equilibria}).

 Given these potential undesirable behaviors, this paper proposes using feed-forward neural network (FNN) to estimate the optimization-based controllers and leveraging reachability analysis for their performance and stability verification in real time. Continuously estimating the reachable sets during the operation of the system has several benefits: \emph{(i)} it provides predictive information on the system's evolution and performance, identifying potential undesirable behaviors such as convergence to undesirable equilibria or loss of controllability~\citep{MFR-APA-PT:21,garg2024advances};  \emph{(ii)} it informs high-level trajectory (re)planners (\cite{matni2024quantitative}), adaptation strategies (\cite{black2023adaptation}), or reach-avoid strategies (\cite{fisac2015reach,landry2018reach}). 

Real-time safety \emph{and} performance verification for optimization-based controllers is a challenging problem (\cite{garg2024advances}) due to several factors: (a) the existing methods are restrictive and often not applicable to the cases where the output of the optimization-based controller cannot be computed in closed form; (b) generating trajectories of these systems can be computationally intractable as one needs to solve a continuum of optimization problems.

\noindent \textbf{Reachability and CBF-based filters}. In the context of safety filters, reachability analysis has been used to refine a CBF that is safe and not overly 
conservative in, e.g.,~\cite{tonkens2022refining}; it has also been used in \cite{tonkens2023patching} with similar offline design purposes in the context of safe value functions. In both cases, Hamilton-Jacobi (HJ) reachability methods were used (\cite{wabersich2023data}). HJ reachability was also used in~\cite{choi2021robust} and~\cite{gong2022constructing} to  construct a control Lyapunov value function to stabilize a point of interest, and in \cite{kumar2023fast} to construct an implicit CBF through HJ reach-avoid differential
dynamic programming. Despite offering high accuracy in estimating reachable sets, HJ-based approaches do not scale well with the size of the system and can become computationally heavy for real-time implementation. Other reachability approaches have been used in~\cite{abate2020enforcing,srinivasan2020continuous} and ~\cite{llanes2022safety} for safety assurance using CBFs. 

\noindent \textbf{Reachability of FNN-controlled systems.}
Reachability of systems with neural network controllers have been studied extensively in the literature. Existing approaches for linear systems include NNV~\citep{HDT-etal:20}, and simulation-guided interval analysis~\citep{WX-HDT-XY-TTJ:21}, ReachLP~\citep{ME-GH-CS-JPH:21}, Reach-SDP~\citep{HH-MF-MM-GJP:20}.
For nonlinear systems, the existing approaches include ReachNN~\citep{CH-JF-WL-XC-QZ:19}, Sherlock~\citep{SD-XC-SS:19}, Verisig 2.0~\citep{RI-TC-JW-RA-GP-IL:21}, POLAR~\citep{CH-JF-XC-WL-QZ:22}, JuliaReach~\citep{CS-MF-SG:22}, and mixed integer programming in~\citep{CS-AM-AI-MJK:22}

\noindent \textbf{Contributions}. We propose a computationally efficient framework for stability and performance verification in systems with optimization-based controllers using reachability analysis. Our approach is based on two main ingredients. First, we approximate the optimal solution map of the optimization problem with a feedforward neural network (FNN) and leverage interval analysis and mixed monotonicity theory to over-approximate the reachable sets of the FNN-approximated system. Second, we present a novel result on the \emph{closeness of trajectories} between the system with the optimization-based controller and the system with the FNN approximation. This result enables us to over-approximate the reachable sets of the original system with the optimization-based controller. 
For the reachability analysis of the approximated system with the FNN, we construct an embedding system by integrating state-of-the-art neural network verification methods with suitable inclusion functions for the open-loop system. This construction of the embedding system is inspired by \cite{jafarpour2023interval}, where interval reachability for safety verification of neural network controlled system was investigated. In this setting, a single trajectory of the embedding system provides a hyper-rectangular over-approximation of the reachable sets for the system with the FNN approximation. We then combine these reachable set estimates with our closeness of trajectories result to obtain an over-approximation of the reachable sets for the original system with optimization-based controller. This over-approximation is expressed as the Minkowski sum of a hyper-rectangle and a ball (defined using an arbitrary norm). 
Our strategy eliminates the need to repeatedly solve optimization problems to compute the controller’s input. Instead, we compute a hyper-rectangular over-approximation of the reachable set for the FNN-approximated system using a single trajectory of the embedding system. Through numerical experiments, we demonstrate that the proposed method is computationally efficient, enabling estimation of multiple reachable sets for the system starting from multiple initial hyper-rectangles, all within a short computation time.

\section{Preliminaries and Problem Statement}

\subsection{Notation and Definitions}

 We denote the set of real numbers, non-negative real numbers, and natural numbers by $\mathbb{R}$, $\mathbb{R}_{\geq 0}$, and $\mathbb{N}$, respectively. Vectors are represented using  lowercase letters $x \in \mathbb{R}^n$, $x_i$ is the $i$th entry of $x$, while matrices use uppercase letters (\emph{e.g.}, $A \in \mathbb{R}^{n \times m}$). Given two sets $\mathcal{X},\mathcal{Y}\in \real^n$, we define the Minkowski sum of the sets $\mathcal{X}$ and $\mathcal{Y}$ by $\mathcal{X}\oplus \mathcal{Y} = \{x+y\mid x\in \mathcal{X},\;\; y\in \mathcal{Y}\}$. Given a norm $\|\cdot\|: \mathbb{R}^n \to \mathbb{R}_{\geq 0}$, the induced norm of a matrix $A$ is denoted by $\|A\|$, and 
 the associated norm ball with radius $r$ centered at $z\in \real^n$ is denoted by $\mathcal{B}_{\|\cdot\|}(r,z)=\{x\in \real^n\mid \|x-z\|\le r\}$. Given a positive definite matrix $P \in \mathbb{R}^{n \times n}$, the weighted $\ell_2$ norm is defined as $\|x\|_P = \sqrt{x^\top P x}$. For a scalar function $g: \mathbb{R}^n \rightarrow \mathbb{R}$, its gradient is denoted by $\nabla g(x)$, and its Hessian matrix by $\nabla^2 g(x)$. For a vector-valued function $g: \mathbb{R}^n \rightarrow \mathbb{R}^m$, the Jacobian matrix is represented by $\frac{\partial g(x)}{\partial x}$.
 A continuous function $\alpha: [0, a) \rightarrow \mathbb{R}_{\geq 0}$ is of class $\mathcal{K}$ if $\alpha(0) = 0$ and it is strictly increasing. A function belongs to class $\mathcal{K}_\infty$ if $a = \infty$ and $\alpha(s) \rightarrow \infty$ as $s \rightarrow \infty$. A continuous function $\alpha: \mathbb{R} \rightarrow \mathbb{R}$ is an extended class-$\mathcal{K}$ function if $\alpha(0) = 0$, it is strictly increasing, and $\alpha(s) \rightarrow \infty$ as $s \rightarrow \infty$. Given a norm $\|\cdot\|: \mathbb{R} \to \mathbb{R}_{\geq 0}$ on $\mathbb{R}^n$ and a continuous map $F: \mathbb{R}^n \to \mathbb{R}$, we denote as $\mathrm{osLip}(f)$ the (minimal) one-sided Lipschitz constant of $f$~\cite[Sec.~3]{FB-CTDS}. Given the vectors $x, z \in \real^n$, $(x, z) \in \real^{2n}$ is their vector  concatenation; i.e., $(x, z) = [x^\top, z^\top]^\top$.  The partial order $\leq$ on $\real^n$ if defined as $x \leq z$ is and only if $x_i \leq z_i$ for all $i = 1, \ldots, n$. For any $x, z \in \real^n$, we define the interval $[x, z] := \{w \in \real^n: x \leq w \leq z \}$. 
For $z, w \in \real^n$ and every $i \in \{1, \ldots, n\}$, we define the vector $z_{[i:w]}$ with entries $z_{[i:w],i} = z_j$ if $j \neq i$ and $z_{[i:w],i} = w_j$ if $j = i$. Finally, given a matrix $A \in \real^{n \times n}$,  we denote the non-negative part of $A$ as $[A]^+ := \max\{A,0\}$   and the non-positive part of $A$ by $[A]^- := \min\{A,0\}$, where the max and min are taken entry-wise.

\vspace{-0.3cm}

\subsection{Main setup and problem statement}

We consider a control-affine dynamical system of the form:

\begin{equation}
\label{eq:control-affine-sys}
    \dot{x} = f(x) + g(x) u,
\end{equation}
where \( x \in \mathbb{R}^n \) is state of the system, \( u \in \mathbb{R}^m \) is the control input, and the functions \( f : \mathbb{R}^n \to \mathbb{R}^n \) and \( g : \mathbb{R}^n \to \mathbb{R}^{n \times m} \) are continuously differentiable. Suppose that a locally-Lipschitz nominal controller $\kappa: \mathbb{R}^n \to \mathbb{R}^m$ is designed so that the system 
$\dot x = \tilde{f}(x) := f(x)+g(x)\kappa(x)$ has a unique equilibrium, and it renders this equilibrium globally asymptotically stable. In the remainder of the paper, we assume without loss of generality that such equilibrium is the origin. We consider the case where it is desirable for the system to operate within a given \emph{safe} set; to this end, in the following we introduce the notion of CBF (see, e.g.,~\cite{ADA-SC-ME-GN-KS-PT:19,xiao2023safe}).

\vspace{.2cm}

\textbf{Definition~1 (CBF)} Let $\mathcal{S}_i\subset\real^n$ be a subset of $\real^n$.
  Let $h_i:\real^{n}\to\real$ be a continuously differentiable function
  such that $\mathcal{S}_i = \{x\in\real^n: h_i(x)\geq0\}$ and $\partial \mathcal{S}_i = \{x\in\real^n: h_i(x)=0\}$.
  The function $h_i$ is a CBF of $\mathcal{S}_i$ for the
  system~\eqref{eq:control-affine-sys} if there exists an extended class
  $\mathcal{K}_{\infty}$ function $\alpha_i$ such that, for all $x \in \mathcal{S}_i$,  $\exists ~ u \in\real^{m}$ satisfying $\nabla h_i(x)^\top (f(x)+g(x)u) + \alpha_i(h_i(x)) \geq 0$. \hfill $\Box$

\vspace{.2cm}
Given $N$ continuously differentiable functions $h_i:\real^{n}\to\real$ defining the sets $\{\mathcal{S}_i\}_{i = 1}^N$ as in Def.~1, we define the safe set for the system~\eqref{eq:control-affine-sys} as $\mathcal{S} := \bigcap_{i=1}^{N} \mathcal{S}_i$. With this definition of CBF and safe set, hereafter we refer to the \emph{filtered} system as:
\begin{equation}\label{eq:general-system-1}
\dot{x}= \tilde{f}(x) +g(x)v(x),
\end{equation}
where the map $x \mapsto v(x)$ is defined as:
\begin{subequations}
\label{eq:optim_filter}
\begin{align}
    v(x) := \arg & \min_{\theta \in \mathbb{R}^m} \frac{1}{2} \| \theta \|^2, \label{eq:filter_cost} \\
    &  \text{s.to:~}  \ell_i(\theta,x) \leq 0 , \quad i = 1, \ldots, L, \label{eq:filter_constraint} \\
    & ~~~~~~~ \nabla h_i(x) ( \tilde{f}(x) + g(x)\theta ) + \alpha_i(h_i(x)) \geq 0,  ~~~ \forall \, i = 1, \ldots N \label{eq:filter_safetyconstraint}
\end{align}
\end{subequations}
with the constraints~\eqref{eq:filter_safetyconstraint} embed the CBFs of the sets $\{\mathcal{S}_i\}_{i = 1}^N$ while the constraints $\ell_i(\theta,x) \leq 0$ specific additional performance and operational requirements; for example,~\eqref{eq:filter_constraint} may include Control Lyapunov Function (CLF) constraints, input constraints, etc. We make the following assumptions. 

\begin{assumption}
\label{as:origin}
For each $i = 1, \ldots, N$, $\nabla h_i(x) \neq 0$ for all $x \in \partial \mathcal{S}_i$. The set  $\mathcal{S} = \bigcap_{i=1}^{N} \mathcal{S}_i$ is non-empty, and the origin is in its interior.  \hfill $\Box$
\end{assumption}

\begin{assumption}
\label{as:regularity}
For any given $x \in \real^n$, each of the functions $\ell_i(\theta,x)$ is convex and continuously differentiable in $\theta$. Additionally, for any given $x \in \real^n$, problem~\eqref{eq:optim_filter} is feasible, and it satisfies the Mangasarian-Fromovitz constraint qualification and the constant-rank condition.      \hfill $\Box$
\end{assumption}

In this setup,~\eqref{eq:optim_filter} is a parametric convex program with a strongly convex cost; hence, it has a unique (globally) optimal solution for any given $x$. If, additionally,  the functions $\ell_i(\theta,x)$ are linear in $\theta$, then~\eqref{eq:optim_filter} is a linearly-constrained quadratic program (QP). Assumption~\ref{as:regularity} ensures that the map $x \mapsto v(x)$  is locally Lipschitz~\cite[Theorem 3.6]{liu1995sensitivity}; see also the conditions in~\cite{chen2024equilibria} for safety filters designed based on only one obstacle. This, in turn, ensures existence and uniqueness of solutions to~\eqref{eq:general-system-1}. Additionally, when Assumptions~\ref{as:origin} and~\ref{as:regularity} hold, it follows from~\cite[Theorem~2]{ADA-SC-ME-GN-KS-PT:19} that the filtered system~\eqref{eq:general-system-1} renders the set $\mathcal{S}$ \emph{forward
invariant}.  

Next, we introduce the notion of \emph{undesirable} equilibrium (also referred to as spurious in, e.g.,~\cite{ MFR-APA-PT:21}). We say that a point $x_{\texttt{unde}}^* \in \real^n$ is an undesirable equilibrium if $\tilde{f}(x_{\texttt{unde}}^*) +g(x_{\texttt{unde}}^*)v(x_{\texttt{unde}}^*) = 0$, and $\tilde{f}(x_{\texttt{unde}}^*) \neq 0$; that is, $x_{\texttt{unde}}^*$ us an equilibrium of the filtered system~\eqref{eq:general-system-1} but not of the system under the nominal controller  $\dot x = \tilde{f}(x)$. We recall that, without loss of generality, the only equilibrium of the  system under the nominal controller is the origin.  

Let $\phi(t,x_0)$ denote the state of the  dynamical system~\eqref{eq:general-system-1} at time $t \geq 0$, when starting from the state $x(0) = x_0 \in \real^n$. Then, given a set of initial states $\mathcal{X}_0 \subset \real^n$, we define the \emph{reachable set} of \eqref{eq:general-system-1} at time $t$ when starting from $\mathcal{X}_0$ as $\mathcal{R}_{\texttt{fs}}(t,\mathcal{X}_0) := \left\{\phi(t,x_0): x_0 \in \mathcal{X}_0  \right\}$, with $t \geq 0$. Since our systems are time-invariant, the notion of reachable set extends to states that are reachable at $t_1  \geq t_0$, when starting from a set of points at time $t_0 \geq 0$. 

As explained in Section~\ref{sec:introduction}, CBF-based controllers 
can introduce undesirable behaviors, including the emergence of spurious equilibria. Continuously estimating the reachable set  of~\eqref{eq:general-system-1} can provide benefits during the \emph{real-time} system operation, including identifying in real time potential undesirable behaviors. Given this, the problem addressed in this paper can be stated as follows.  

\vspace{.2cm}

\textbf{Problem~1}. For the system~\eqref{eq:general-system-1}, develop a computationally efficient method to estimate a reachable set  $\mathcal{R}_{\texttt{fs}}(t,\mathcal{X})$ (or multiple reachable sets  $\mathcal{R}_{\texttt{fs}}(t,\mathcal{X}_i)$, $i = 1, \ldots, N_s$). \hfill $\Box$ 

\vspace{.2cm}

In particular, we are interested in methods with low computational complexity, so that reachable sets $\mathcal{R}_{\texttt{fs}}(t,\mathcal{X})$ can be estimated efficiently and they can be continuously updated. 

\vspace{-0.3cm}
\section{Reachability via Neural Networks and Mixed-monotonicity}
\label{sec:reachability}

\subsection{Main framework}

We address Problem~1 by proposing a framework to estimate the reachable sets for the filtered system~\eqref{eq:general-system-1} by: \emph{(i)}  approximating  $v(x)$ using a pre-trained FNN; \emph{(ii)} leveraging mixed monotonicity to efficiently compute approximations of reachable set of the system when $v(x)$ is approximated by  the FNN; \emph{(iii)} leveraging closeness of trajectory results to compute over-approximations of $\mathcal{R}_{\texttt{fs}}(t,\mathcal{X}_0)$. 

\textbf{Neural network approximation}. We approximate $v(x)$ using a FNN $\mathcal{N}(x)$, trained offline to replicate optimal solutions to  \eqref{eq:optim_filter}. Here, $\mathcal{N}(x)$ is structured as an \( H \)-layer FNN, defined as:
\begin{subequations}
    \label{eq:NN}
\begin{align}
u & = \mathcal{N}(x) := W^{(H)} \varphi^{(H)} +   b^{(H)} \\
    \varphi^{(i)} &= \Phi^{(i)} \left( W^{(i-1)} \varphi^{(i-1)} + b^{(i-1)} \right), \quad i = 1, \dots, H,  ~~ \text{and}~~\varphi^{(0)}  = x \end{align}
\end{subequations}
where \( W^{(i-1)} \in \mathbb{R}^{n_i \times n_{i-1}} \) and \( b^{(i-1)} \in \mathbb{R}^{n_i} \) are the weight matrix and bias vector of the $i$th layer of the  network, $n_i$ is the number of neurons in the $i$th layer, $\varphi^{(i)} \in \real^{n_i}$ is the $i$th hidden variable, and $\Phi^{(i)}: \real^{n_i} \to \real^{n_i}$ is the Lipschitz-continuous  diagonal  activation function of the $i$th layer. We assume that each element satisfies the inequalities $0\leq \frac{(\Phi^{(i)}_\ell(x) - \Phi^{(i)}_\ell(y))}{x - y} \leq 1$, for any pair $x, y$, and for $\ell = 1, \ldots, n_i$. We note that activation functions such as the ReLU, leaky ReLU, sigmoid, and tanh, all satisfy this condition. The dynamics under the FNN approximation of the controller  are then given by:
\begin{equation}
\label{eq:general-system-N}
    \dot{x}_\texttt{nn} = \tilde{f}(x_\texttt{nn}) + g(x_\texttt{nn}) \mathcal{N}(x_\texttt{nn})
\end{equation}
where we recall that $\tilde{f}(x_\texttt{nn}) = f(x_\texttt{nn}) + g(x_\texttt{nn})\kappa(x_\texttt{nn})$, and where the subscript $\texttt{nn}$ is used to emphasize that trajectories are generated using an approximation of $v(x)$. Hereafter, similarly to~\eqref{eq:general-system-1}, we let $\phi_{\texttt{nn}}(t,x_0)$ denote the state of the dynamical system~\eqref{eq:general-system-N} at time $t \geq 0$, when starting from the state $x(0) =x_0 \in \real^n$. Additionally, we define the reachable set of \eqref{eq:general-system-N} at time $t$ when starting from $\mathcal{X}_0$ as $\mathcal{R}_{\texttt{nn}}(t,\mathcal{X}_0) := \left\{\phi_{\texttt{nn}}(t,x_0: x_0 \in \mathcal{X}_0  \right\}$, with $t \geq 0$.

\textbf{Embedding system via inclusion functions}. In the next step, we review the framework proposed in~\citep{SJ-AH-SC:24} for interval reachability of the approximate system~\eqref{eq:general-system-N}. The main idea is to embed the neural network controlled system~\eqref{eq:general-system-N} into a larger dimensional space~\citep[Theorem 3]{SJ-AH-SC:24} and use a single trajectory of this embedding system to over-approximate reachable sets of the original system~\eqref{eq:general-system-N}. For a given function $d:\real^n\to \real^m$ and an interval $[\underline{y},\overline{y}]$, we introduce the notion of inclusion function of $d$ on interval $[\underline{y},\overline{y}]$ as the map $\OD_{[\underline{y},\overline{y}]} = \left[\begin{smallmatrix}
    \underline{\OD}_{[\underline{y},\overline{y}]}\\ \overline{\OD}_{[\underline{y},\overline{y}]}
\end{smallmatrix}\right]: \real^{2n}\to \real^{2m}$ satisfying~\citep{LJ-MK-OD-EW:01} 
\begin{align*}
    \underline{\OD}_{[\underline{y},\overline{y}]}(\underline{x},\overline{x}) \le d(x) \le \overline{\OD}_{[\underline{y},\overline{y}]}(\underline{x},\overline{x}),\qquad\mbox{for all }x\in [\underline{x},\overline{x}]\subseteq [\underline{y},\overline{y}]. 
\end{align*}
For a given function $d:\real^{n}\to \real^m$, there exists various methods to compute its inclusion function over a given interval including natural compositional approach~\citep{LJ-MK-OD-EW:01}, Jacobian-based methods, and decomposition-based methods~\citep{SC-MA:15b}. We refer to~\citep{SJ-AH-SC:24} for more details on construction of inclusion functions. 
For a FNN $\mathcal{N}:\real^n\to \real^m$, one can alternatively obtain the inclusion function for $\mathcal{N}(x)$ using existing neural network verification algorithms including CROWN~\citep{HZ-etal:18}, LipSDP \citep{MF-MM-GJP:22}, and IBP \citep{SG-etal:19}. In particular, some neural network verification algorithms can provide \textit{affine} $[y,\widehat{y}]$-localized inclusion functions for $\mathcal{N}$; examples include CROWN and its variants~\citep{HZ-etal:18}. Given an interval $[y,\widehat{y}]$, these algorithms provide a tuple $(C,\underline{d},\overline{d})$ defining affine upper and lower bounds for the output of the neural network
\begin{gather} \label{eq:crown}
   C(y,\widehat{y})x + \underline{d}(y,\widehat{y}) \leq \mathcal{N}(x) \leq C(y,\widehat{y})x + \overline{d}(y,\widehat{y}),
\end{gather}
for every $x\in [y,\widehat{y}]$. Consider the FFN-controlled system~\eqref{eq:general-system-N} with an inclusion function $\tilde{\OF} = \left[\begin{smallmatrix}
    \tilde{\underline{\OF}}\\
    \tilde{\overline{\OF}}
\end{smallmatrix}\right]:\real^{2n}\to \real^{2n}$ for the vector field $\tilde{f}$ and an inclusion function $\OG = \left[\begin{smallmatrix}
   \underline{\OG}\\
   \overline{\OG}
\end{smallmatrix}\right]:\real^{2n}\to \real^{2m}$ for the map $g$. We can construct the embedding system associated to the FNN-controlled system~\eqref{eq:general-system-N} as follows: 
\begin{align}\label{eq:embeding}
    \dot{\overline{x}}_i &= \tilde{\underline{\OF}}(\underline{x},\overline{x}_{[i:\underline{x}]}) + [\overline{\OG}(\underline{x},\overline{x})]^{+}(C_{[\underline{x},\overline{x}]}x + \underline{d}_{[\underline{x},\overline{x}]}) + [\overline{\OG}(\underline{x},\overline{x})]^{-}(C_{[\underline{x},\overline{x}]}x + \overline{d}_{[\underline{x},\overline{x}]}) \nonumber\\ 
    \dot{\overline{x}}_i &= \tilde{\overline{\OF}}(\underline{x}_{[i:\overline{x}]},\overline{x}) + [\overline{\OG}(\underline{x},\overline{x})]^{-}(C_{[\underline{x},\overline{x}]}x + \underline{d}_{[\underline{x},\overline{x}]}) + [\overline{\OG}(\underline{x},\overline{x})]^{+}(C_{[\underline{x},\overline{x}]}x + \overline{d}_{[\underline{x},\overline{x}]}) \, .
\end{align}
Let $\mathcal{X}_0= [\underline{x}_0,\overline{x}_0]$ and $t\mapsto \left[\begin{smallmatrix}
    \underline{x}(t)\\ \overline{x}(t)
\end{smallmatrix}\right]$ be a trajectory of the embedding system~\eqref{eq:embeding} starting from $\left[\begin{smallmatrix}
    \underline{x}(0)\\ \overline{x}(0)
\end{smallmatrix}\right] = \left[\begin{smallmatrix}
    \underline{x}_0\\ \overline{x}_0
\end{smallmatrix}\right]$. Then, by~\cite[Theorem 4(i)]{SJ-AH-SC:24}, it follows that for every $t\ge 0$, 
\begin{align}\label{eq:interval-over}
    \mathcal{R}_{\texttt{nn}}(t,\mathcal{X}_0) \subseteq [\underline{x}(t),\overline{x}(t)]. 
\end{align}

\subsection{Closeness of solutions and over-approximation of the reachable sets}

In this section, we obtain bounds on the distance between the trajectories of the filtered system~\eqref{eq:general-system-1} and its approximation using the FNN-controlled system~\eqref{eq:general-system-N}. We begin with the following assumption on the accuracy of the FNN $x\mapsto \mathcal{N}(x)$ in approximating the map $x\mapsto v(x)$ defined in~\eqref{eq:optim_filter}. 

\begin{assumption}[FNN approximation error]\label{eq:nnclose} For every norm $\|\cdot\|_{\mathcal{U}}$ on $\real^m$ and every compact set $\mathcal{C}\subseteq \mathbb{R}^n$, there exists a constant $M_{\mathcal{C}}\ge 0$ such that $\|\mathcal{N}(x)-v(x)\|_{\mathcal{U}}\le M_{\mathcal{C}}$, for every $x\in \mathcal{C}$. \hfill $\Box$
\end{assumption}

This assumption is grounded on the universal approximation capabilities of neural networks on compact sets (see the seminal works of~\cite{hornik1991approximation, barron1992neural} and some recent results in, for example,~\cite{duan2023minimum, wang2024universal}). One can use Assumption~\ref{eq:nnclose} and the notion of incremental state-to-input bounds to obtain upper bounds on the distance between the trajectories of the filtered system~\eqref{eq:general-system-1} and their associated trajectories of the FNN-controlled system~\eqref{eq:general-system-N}. 

\begin{theorem}[Closedness of solutions]\label{thm:closed}
    Consider the filtered system~\eqref{eq:general-system-1} and its approximation~\eqref{eq:general-system-N} with the FNN $\mathcal{N}(x)$ satisfying Assumption~\ref{eq:nnclose}. Suppose that $t\mapsto \left[\begin{smallmatrix}
        \underline{x}_{\texttt{nn}}(t)\\ \overline{x}_{\texttt{nn}}(t)
    \end{smallmatrix}\right]$ is a curve such that $\mathcal{R}_{\texttt{nn}}(t,\mathcal{X}_0)\subseteq [\underline{x}_{\texttt{nn}}(t),\overline{x}_{\texttt{nn}}(t)]$, $\forall ~ t\ge 0$, $\mathcal{X}_0= [\underline{x}_0,\overline{x}_0]$. Let $t\mapsto x(t)$ be a trajectory of the filtered system~\eqref{eq:general-system-1} and $t\mapsto x_{\texttt{nn}}(t)$ be associated the trajectory of the neural network controlled system~\eqref{eq:general-system-N}, for $x(0)=x_{\texttt{nn}}(0)$, and $\mathcal{C}\subseteq \real^n$ be a compact set that contains both trajectories.  Then, $\forall ~ t\ge 0$,
    \begin{align}
    \label{eq:closeness}
        \|x(t)-x_{\texttt{nn}}(t)\| \le \ell M_{\mathcal{C}} \exp\left(\int_{0}^{t} c(\tau)d\tau\right) \int_{0}^{t} \exp\left(-\int_{0}^{\tau}c(s)ds\right)d\tau
    \end{align}
    where $\ell = \sup_{x\in \mathcal{C}} \|g(x)\|_i$  and $\|\cdot\|_i$ is induced norm from $\|\cdot\|$ and $\|\cdot\|_{\mathcal{U}}$ and $t\mapsto c(t)$ is a curve satisfying 
    \begin{align}\label{eq:condition}
        \mathrm{osLip}\left(\tilde{f}(x) + C_{[\underline{x}_{\texttt{nn}}(t),\overline{x}_{\texttt{nn}}(t)]} x\right) \le c(t),\qquad\mbox{ for every }x\in \mathcal{C}, \;\; t\ge 0 \, .
    \end{align}
\end{theorem}

\begin{proof}
    We first note that, for every $t\ge 0$ and every $x\in [\underline{x}_{\texttt{nn}}(t),\overline{x}_{\texttt{nn}}(t)]$, we have
    \begin{align*}   C_{[\underline{x}_{\texttt{nn}}(t),\overline{x}_{\texttt{nn}}(t)]}x + \underline{d}_{[\underline{x}_{\texttt{nn}}(t),\overline{x}_{\texttt{nn}}(t)]} \le \mathcal{N}(x)\le C_{[\underline{x}_{\texttt{nn}}(t),\overline{x}_{\texttt{nn}}(t)]}x + \overline{d}_{[\underline{x}_{\texttt{nn}}(t),\overline{x}_{\texttt{nn}}(t)]}
    \end{align*}
    This implies that $\mathcal{N}(x) = C_{[\underline{x}_{\texttt{nn}}(t),\overline{x}_{\texttt{nn}}(t)]}x + d(t)$, where $d(t) \in [\underline{d}_{[\underline{x}_{\texttt{nn}}(t),\overline{x}_{\texttt{nn}}(t)]},\overline{d}_{[\underline{x}_{\texttt{nn}}(t),\overline{x}_{\texttt{nn}}(t)]}]$. Now consider the following control-affine system $\dot{x} = \tilde{f}(x) + g(x) \mathcal{N}(x) + g(x) u$. 
    For $u(t) = 0$, the trajectory of the system starting from $x(0)$ is $t\mapsto x_{\texttt{nn}}(t)$. For $u(t) = v(x(t)) - \mathcal{N}(x(t))$, the trajectory of the system starting from $x(0)$ is $t\mapsto x(t)$. Letting $D^+$ denote the Dini upper right-hand derivative, we have that (see, e.g.,~\cite[Thm. 3.16]{FB-CTDS})
    \begin{align*}
        D^{+}\|x(t)-x_{\texttt{nn}}(t)\| & \le c(t) \|x(t)-x_{\texttt{nn}}(t)\| + \ell \|v(x(t)) - \mathcal{N}(x(t))\| \\
        & \leq c(t) \|x(t)-x_{\texttt{nn}}(t)\| + \ell M_{\mathcal{C}}
    \end{align*}
 for all $t \geq 0$, where we used Assumption~\ref{eq:nnclose} for the second inequality.  For notational simplicity, let $\xi(t) := \exp \left(\int_0^t c(s) d s\right)$, and note that $\xi(t) \geq 0$ for any $t \geq 0$. Then, we compute 
 \begin{align*}
 D^{+} \left(\frac{\|x(t)-x_{\texttt{nn}}(t)\|}{\xi(t)} \right) & = \frac{\left(D^{+}\|x(t)-x_{\texttt{nn}}(t)\| \right) \xi(t) - \dot{\xi}(t)\|x(t)-x_{\texttt{nn}}(t)\|} {\xi^2(t)} \\
 & \leq \frac{(c(t) \|x(t)-x_{\texttt{nn}}(t)\| + \ell M_{\mathcal{C}}) \xi(t) - c(t) \xi(t) \|x(t)-x_{\texttt{nn}}(t)\|}{\xi^2(t)} \leq \frac{\ell M_{\mathcal{C}}}{\xi(t)} \, .
\end{align*}
Next, integrating over time, we get $\frac{\|x(t)-x_{\texttt{nn}}(t)\|}{\xi(t)} \leq \frac{\|x(0)-x_{\texttt{nn}}(0)\|}{\xi(0)} + \int_0^t  \frac{\ell M_{\mathcal{C}}}{\xi(s)} d s$, where we note that $\|x(0)-x_{\texttt{nn}}(0)\| = 0$ since $x(0) = x_{\texttt{nn}}(0)$. Therefore, 
\begin{align*}
  \|x(t)-x_{\texttt{nn}}(t)\| \leq \ell M_{\mathcal{C}} \xi(t)  \int_0^t  \tfrac{1}{\xi(s)} d s = \ell M_{\mathcal{C}} \exp\left(\int_{0}^{t} c(\tau)d\tau\right) \int_{0}^{t} \exp\left(-\int_{0}^{\tau}c(s)ds\right)d\tau \, . 
\end{align*}
\end{proof}

Using this theorem, we can find over-approximation of the reachable set of~\eqref{eq:general-system-1} using the hyper-rectangular over-approximations of reachable sets of the FNN-controlled system~\eqref{eq:general-system-N} given by~\eqref{eq:interval-over}.  

\begin{theorem}[Minkoswki Sum]\label{eq:mink}
    Consider the filtered system~\eqref{eq:general-system-1} and its approximation~\eqref{eq:general-system-N} with the FNN $\mathcal{N}(x)$ satisfying Assumption~\ref{eq:nnclose}. Suppose that $t\mapsto \left[\begin{smallmatrix}
        \underline{x}_{\texttt{nn}}(t)\\ \overline{x}_{\texttt{nn}}(t)
    \end{smallmatrix}\right]$ is a curve such that $\mathcal{R}_{\texttt{nn}}(t,\mathcal{X}_0)\subseteq [\underline{x}_{\texttt{n}}(t),\overline{x}_{\texttt{nn}}(t)]$, for every $t\ge 0$, with $\mathcal{X}_0= [\underline{x}_0,\overline{x}_0]$. Assume that the reachable sets $\mathcal{R}_{\texttt{fs}}(t,\mathcal{X}_0)$ and $\mathcal{R}_{\texttt{nn}}(t,\mathcal{X}_0)$ are bounded and $\mathcal{C}\subseteq \real^n$ is a compact set such that $ \mathcal{R}_{\texttt{fs}}(t,\mathcal{X}_0)\cup  \mathcal{R}_{\texttt{nn}}(t,\mathcal{X}_0)\subseteq \mathcal{C}$ for all times $t\ge 0$. Then, 
    \begin{align}\label{eq:minkowski}
    \mathcal{R}_{\texttt{fs}}(t,\mathcal{X}_0) \subseteq  [\underline{x}_{\texttt{nn}}(t), \overline{x}_{\texttt{nn}}(t)] \oplus \mathcal{B}_{\|\cdot\|}(r_t,0),\qquad\mbox{for all } t\ge 0,
\end{align}
where $\oplus$ is the Minkowski sum and $r_t = \ell M_{\mathcal{C}} \exp\left(\int_{0}^{t} c(\tau)d\tau\right) \int_{0}^{t} \exp\left(-\int_{0}^{\tau}c(s)ds\right)d\tau$ with $\ell = \sup_{x\in \mathcal{C}} \|g(x)\|_{i}$  where $\|\cdot\|_{i}$ is the induced norm from $\|\cdot\|$ and $\|\cdot\|_{\mathcal{U}}$ and $t\mapsto c(t)$ is a curve satisfying the inequality~\eqref{eq:condition}
\end{theorem}
\begin{proof}
    The result follows by combining Theorem~\ref{thm:closed} and the hyper-rectangular bound~\eqref{eq:interval-over} and the definition of Minkowski sum. 
\end{proof}

Computing the Minkowski sum of two arbitrary sets can be computationally complicated. However, when the sets are ellipsoids or polytopes, there exist efficient algorithms for estimating their Minkowski sum~\citep{PG-BS:93,AH:18}. In particular, when $\mathcal{B}_{\|\cdot\|}(r_{t},0) = \mathcal{B}_{\infty}(r_t,0)$ is $\ell_\infty$-norm ball, the Minkowski sum~\eqref{eq:minkowski} can be easily computed as $[\underline{x}_{\texttt{nn}}(t),\overline{x}_{\texttt{nn}}(t)] \oplus \mathcal{B}_{\infty}(r_t,0) = [\underline{x}_{\texttt{nn}}-r_t, \overline{x}_{\texttt{nn}}+r_t]$. 

\vspace{-0.3cm}
\section{Numerical Experiments}

We show the efficiency of our proposed reachability method for performance verification of CBF-based controllers with two sets of numerical experiments.

\vspace{.1cm}

\noindent \textbf{Scenario~1} \emph{(Integrator dynamics with circular obstacle)}. Consider the integrator dynamics $\dot{x} = u$,  where $x \in \mathbb{R}^2$ represents the state vector and $u \in \mathbb{R}^2$ is the control input. The nominal controller is designed as $\kappa(x) = K x$, with $K = \begin{bmatrix} -1 & 0 \\ 0 & -5 \end{bmatrix}$, and it stabilizes the origin.  We consider a circular obstacle located at $\boldsymbol{o} = [2, 0]^\top$ with a radius of $r = 1$, and define the safe set $\mathcal{S} = \{ x \in \mathbb{R}^2 \mid h_1(x) \geq 0 \}$, with CBF $h_1(x) = \| x - \boldsymbol{o} \|^2 - r^2$. By analyzing the closed-loop dynamics with the CBF filter, from~\cite{chen2024equilibria} it follows that  additional equilibria emerge at $x_{\texttt{unde}}^* = \left( \frac{5}{2}, \pm \frac{\sqrt{3}}{2} \right)^\top$ and $x_{\texttt{unde}}^* = [3, 0]^\top$. Evaluating the Jacobian $J_{h_1,\alpha}(x_{\texttt{unde}}^*)$ at these points, it follows that that the equilibrium at $[3, 0]^\top$ is locally asymptotically stable, while the other two are saddle points. To  better highlight the stability of the undesirable equilibrium, we introduce the modified CBF (a scaled, equivalent version of $h_1$) $h_2(x) = \left( \| x - \boldsymbol{o}_2 \|^2 + 1 \right) h_1(x)$, where $\boldsymbol{o}_2 = [5, 1]^\top$~\cite[Sec.~7]{chen2024equilibria}, and we use $\alpha(s) = s$. The safety filter computes the control input $v(x)$ by using $h_2(x)$ and without additional constraints. We train a FNN offline to approximate the mapping $x \mapsto v(x)$. The neural network is structured as an FNN with architecture [2 $\times$ 400 $\times$ 400 $\times$ 400 $\times$ 2, ReLU activations], resulting in the closed-loop approximate system $\dot{x}_{\texttt{nn}} = \kappa(x_{\texttt{nn}}) + \mathcal{N}(x_{\texttt{nn}})$, as per~\eqref{eq:general-system-N}. The network was trained using 99,472 data points over the region illustrated in the Figure~\ref{fig:scenario1}, with points that are equi-spaced on a grid. The training achieved 
a maximum $\ell_2$-norm error of $0.2$ and a maximum $\ell_\infty$-norm  error of $0.19$. We simulate the system dynamics using Euler integration with a step size of $0.04$.

 \begin{figure}[t!]
  \centering
  \includegraphics[width=14.0cm]{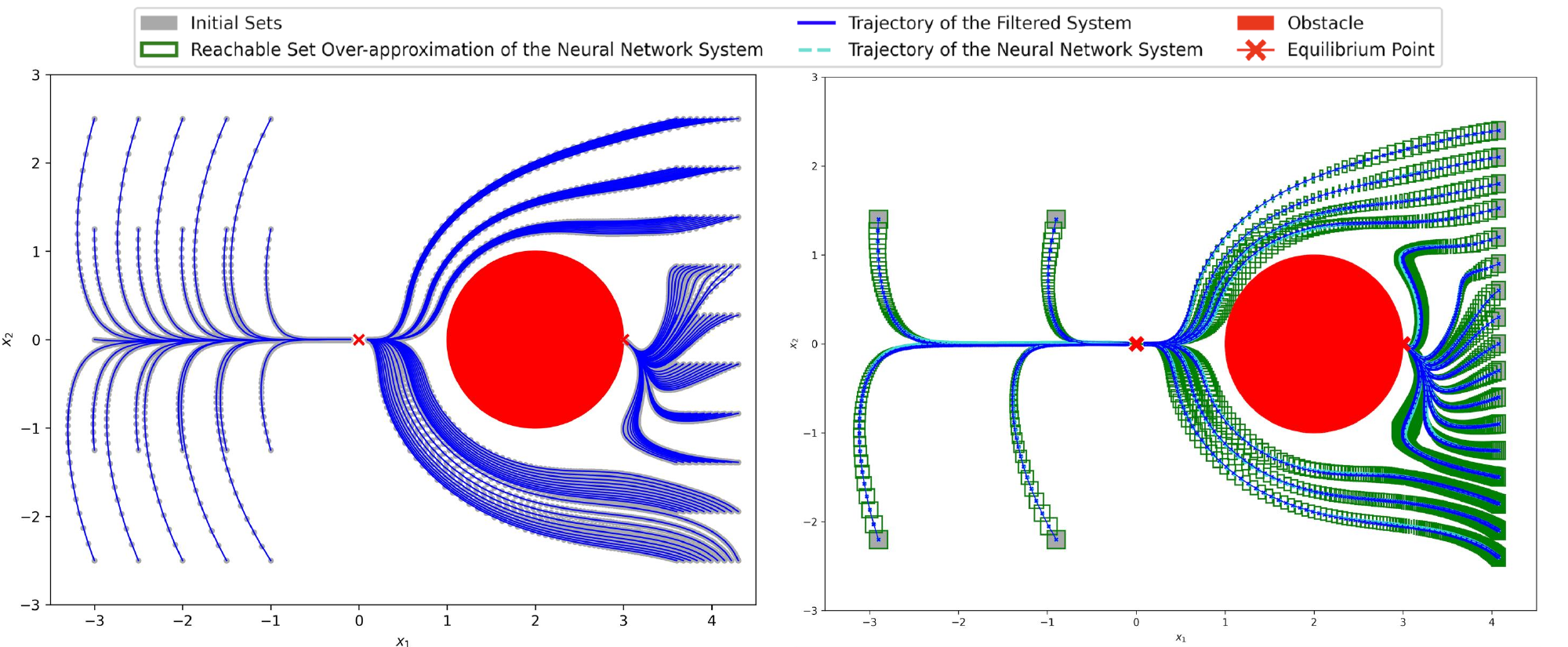}
  \vspace{-.2cm}
  \caption{(\textbf{Left}) Trajectories of~\eqref{eq:general-system-1} in Scenario 1.  (\textbf{Right}) System~\eqref{eq:general-system-N} with the FNN approximation, and hyper-rectangles $[\underline{x}_{\texttt{nn}}(t), \overline{x}_{\texttt{nn}}(t)]$ starting from several initial sets $[\underline{x}_{\texttt{nn}}(0), \overline{x}_{\texttt{nn}}(0)]$.   
  }
  \label{fig:scenario1}
  \vspace{-.4cm}
\end{figure}

\begin{figure}[t!]
  \centering
  \includegraphics[width=14.0cm]{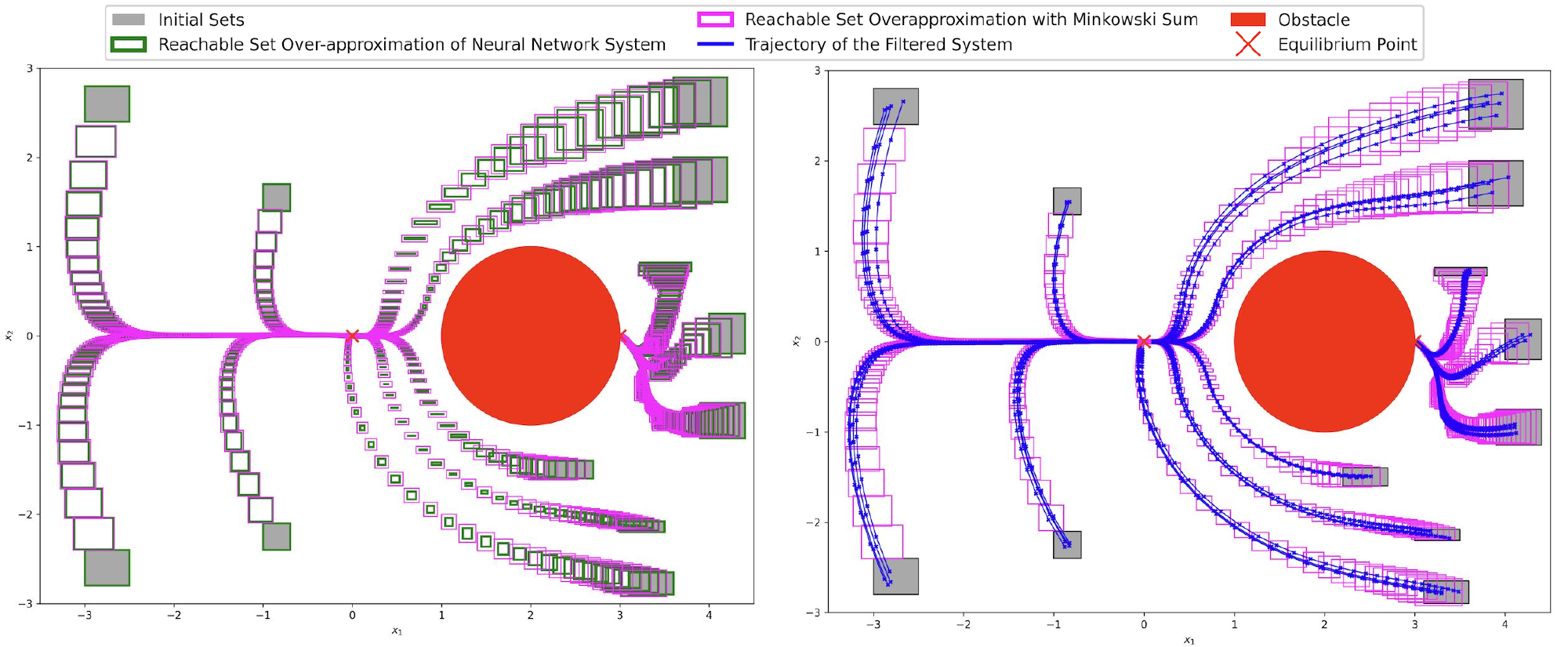}
  \vspace{-.2cm}
  \caption{(\textbf{Left}) Over-approximations  $[\underline{x}_{\texttt{nn}}(t), \overline{x}_{\texttt{nn}}(t)]$ and $[\underline{x}_{\texttt{nn}}(t), \overline{x}_{\texttt{nn}}(t)] \oplus \mathcal{B}(r_t,0)$ of the reachable sets, for different initial sets, over $T = 10$ seconds. (\textbf{Right}) As expected, trajectories of the filtered system ~\eqref{eq:general-system-1} are contained in  
  $[\underline{x}_{\texttt{nn}}(t), \overline{x}_{\texttt{nn}}(t)] \oplus \mathcal{B}_{\infty}(r_t,0)$.}
  \label{fig:scenario12}
  \vspace{-.6cm}
\end{figure}

Figure~\ref{fig:scenario1} (left) shows the trajectories of the filtered system for several initial points, over a time horizon of $T = 10$ seconds. Several trajectories converge to the undesirable equilibrium at $[3, 0]^\top$, located on the obstacle boundary. Figure~\ref{fig:scenario1} (right) shows trajectories $x_{\texttt{nn}}(t)$ of the system with the FNN approximation, along with the hyper-rectangles $[\underline{x}_{\texttt{nn}}(t_k), \overline{x}_{\texttt{nn}}(t_k)]$, for several times $t_k$ and for several sets of initial states $\mathcal{X}_0 = [\underline{x}_0, \overline{x}_0]$ (gray rectangles), over a time horizon of $T = 10$ seconds. The affine inclusion functions for the FNN are obtained from CROWN and computed using autoLiRPA, and we construct the embedding system as described in Section~\ref{sec:reachability}. Figure~\ref{fig:scenario12} (left) illustrates the sets $[\underline{x}_{\texttt{nn}}(t), \overline{x}_{\texttt{nn}}(t)]$ and $[\underline{x}_{\texttt{nn}}(t), \overline{x}_{\texttt{nn}}(t)] \oplus \mathcal{B}_{\infty}(r_t,0)$, constructed as described in Section~\ref{sec:reachability}, for several initial sets $\mathcal{X}_0 = [\underline{x}_0, \overline{x}_0]$; here, $\mathcal{B}_{\infty}(r_t, 0)$ is computed with respect to $\ell_\infty$-norm using Theorem~\ref{thm:closed} with $\mathcal{C}= [-3.5,4.5]\times [-3,3]$ and $M_{\mathcal{C}} = 0.19$. Figure~\ref{fig:scenario12} (left) shows that trajectories fo the filtered system~\eqref{eq:general-system-1} are contained in the sets  $[\underline{x}_{\texttt{nn}}(t), \overline{x}_{\texttt{nn}}(t)] \oplus \mathcal{B}_{\infty}(r_t,0)$. For the initial sets shown in Figure~\ref{fig:scenario12}, our reachability framework efficiently and rigorously verifies whether the system converges to the origin or a spurious equilibrium point. The time to compute $[\underline{x}_{\texttt{nn}}(t), \overline{x}_{\texttt{nn}}(t)]$ in  Figure~\ref{fig:scenario12} from the 12 initial sets across all was approximately 2.15 seconds, and the time to compute the Minkowski sum was approximately 0.02 sec.

\vspace{.1cm}

\noindent \textbf{Scenario~2} \emph{(Double integrator dynamics with multiple obstacles)}. Next, we extend the experimental results to  multiple obstacles. We consider a double integrator system with position $x_1$, velocity $x_2$, and an acceleration control input $u \in \mathbb{R}$. The system dynamics are expressed as $\dot{x}_1 = x_2$ and $\dot{x}_2 = u$, and the nominal control input is designed as $\kappa(x) = - k^\top x$, with $k = [1, 2]^\top$, stabilizing the closed-loop system to the origin. We define two CBF constraints, associated with two unsafe regions.
The first constraint prevents collisions with a circular obstacle located at $\boldsymbol{o} = [2, 0]^\top$ with radius $r_1 = 1$. The corresponding safe set is defined as: $\mathcal{S}_1 = \{x \in \mathbb{R}^2 \mid h_1(x) \geq 0\}, \quad h_1(x) = (x_1 - 2)^2 + x_2^2 - r_1^2$. The second constraint ensures that the system remains in the half-plane $x_2 \geq \bar{x}_2 = -2$; the corresponding safe set is given by: $\mathcal{S}_2 = \{x \in \mathbb{R}^2 \mid h_2(x) \geq 0\}, \quad h_2(x) = x_2 - \bar{x}_2$. These constraints ensure that the overall safe set $\mathcal{S} = \bigcap_{i=1}^2 \mathcal{S}_i$ remains forward invariant. To approximate the optimization-based filter $v(x)$, we use 91,326 data points uniformly distributed across a grid within the region shown in Figure~\ref{fig:scenario2} to train a FNN with the architecture [ [2 $\times$ 500 $\times$ 500 $\times$ 500 $\times$ 1], ReLU activations]. The network achieved a maximum $\ell_\infty$-norm error of $0.26$. The embedding system~\eqref{eq:embeding}  is computed using FNN inclusion functions from CROWN and simulated with Euler integration (step size 0.05, simulation time $T = 10$ seconds). We use Theorem~\ref{eq:mink} where $\mathcal{B}_{\infty}(r_t,0)$ is computed with respect to $\ell_\infty$-norm with $\mathcal{C} = [-4,6]\times [-2,2]$ and $M_{\mathcal{C}} = 0.26$. In this scenario, $v(x)$ cannot be computed in closed form, further motivating our approach. For the 9 initial sets shown in Figure~\ref{fig:scenario2}, our reachability framework rigorously verifies that the system converges to the origin. The time to compute $[\underline{x}_{\texttt{nn}}(t), \overline{x}_{\texttt{nn}}(t)]$ from the 9 initial sets was approximately 2.7 sec. on average, and the Minkowski sum computation took approximately 0.015 sec. 

\begin{figure}
  \centering
  \includegraphics[width=14.0cm]{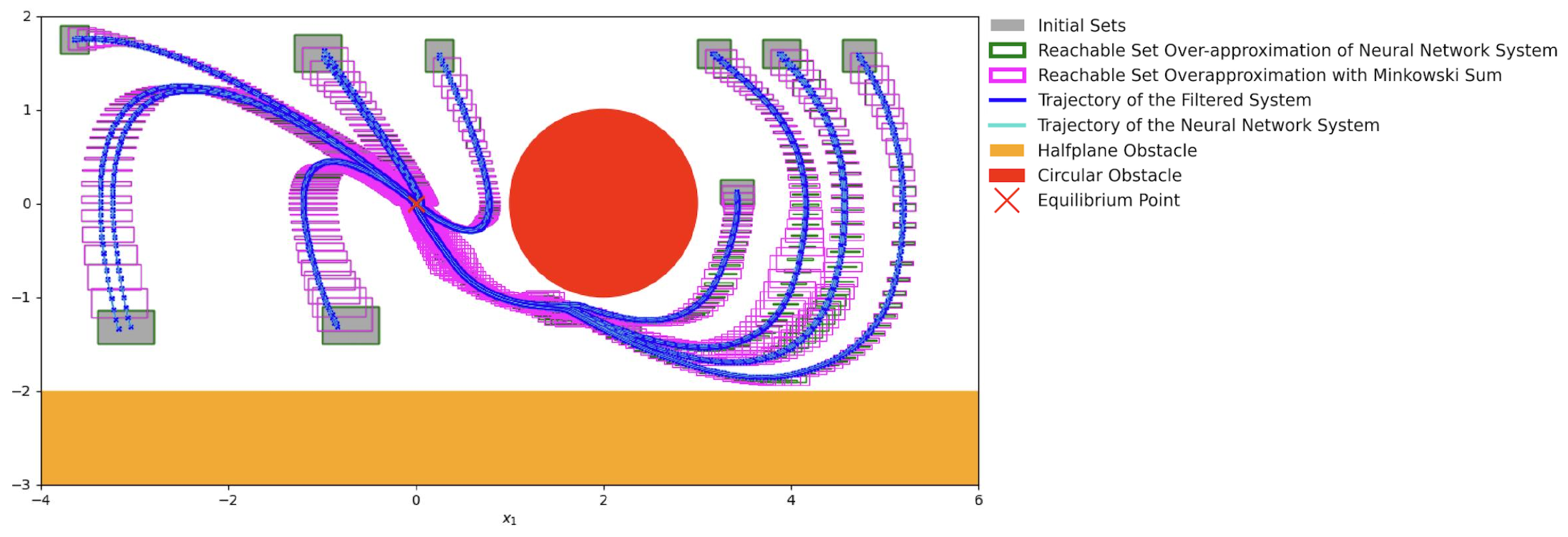}
  \vspace{-.3cm}
  \caption{Trajectories of~\eqref{eq:general-system-1} in Scenario 2, hyper-rectangles  $[\underline{x}_{\texttt{nn}}(t), \overline{x}_{\texttt{nn}}(t)]$, and over-approximations $[\underline{x}_{\texttt{nn}}(t), \overline{x}_{\texttt{nn}}(t)] \oplus \mathcal{B}_{\infty}(r_t,0)$ of reachable sets, for different initial sets.
  }
  \label{fig:scenario2}
  \vspace{-.6cm}
\end{figure}

\vspace{-0.4cm}

\section{Conclusions}

We proposed a computationally efficient interval reachability method for performance verification of systems with optimization-based controllers. Our method builds on an approximation of the optimization-based controller by a pre-trained FNN, the use of an embedding system, and bounds for the distance of solutions of the system with the optimization-based controller and the FNN approximation.  For future research, we plan to extend the proposed approach to systems with disturbance, and investigate reachability-based control design approaches.

\acks{This work was supported in part by the National Science Foundation award 2448264, by the Air Force Office of Scientific Research award FA9550-23-1-0740, and by the U.S. Department of Energy, Office of Electricity, Advanced Grid Modeling program.}

\bibliography{biblio.bib}

\end{document}